\documentclass[aps,pra,superscriptaddress,twocolumn,10pt]{revtex4-2}
\usepackage[colorlinks=true, citecolor=blue, urlcolor=blue ]{hyperref}
\usepackage{epsfig,newlfont,amssymb,amsfonts,amsmath,bm,subfigure,palatino,mathtools,amsthm,braket,soul,enumitem,color,graphics,graphicx,times,physics,makecell,multirow,float}
\usepackage[normalem]{ulem}
\usepackage{multirow}
\usepackage{braket}
\newtheorem{theorem}{Theorem}
\newcolumntype{P}[1]{>{\centering\arraybackslash}p{#1}}
\graphicspath{{Figures/}{Dynamics/}}

\usepackage{natbib}

\begin{document}
\title{
Efficient Estimation of Multiple Temperatures via a Collisional Model
}
\author{Srijon Ghosh} 
\email{srijon.ghosh@unipa.it}
\affiliation{Dipartimento di Ingegneria, Università degli Studi di Palermo, Viale delle Scienze, 90128 Palermo, Italy}

\author{Sagnik Chakraborty}
\email{sagnik.chakraborty@unipa.it}
\affiliation{Dipartimento di Ingegneria, Università degli Studi di Palermo, Viale delle Scienze, 90128 Palermo, Italy}

\author{Rosario Lo Franco}
\email{rosario.lofranco@unipa.it}
\affiliation{Dipartimento di Ingegneria, Università degli Studi di Palermo, Viale delle Scienze, 90128 Palermo, Italy}

\begin{abstract}

We present a quantum thermometric protocol for the estimation of multiple temperatures within the collisional model framework. Employing the formalism of multiparameter quantum metrology, we develop a systematic strategy to estimate the temperatures of several thermal reservoirs with minimal estimation error. We prove a necessary and sufficient condition for the singularity of the Fisher information matrix for a bi-parametrized qubit state. By using controlled rotations of ancillary systems between successive interaction stages, we eliminate parameter interdependencies, thereby rendering the quantum Fisher information matrix non-singular. Remarkably, we demonstrate that precision enhancement in the joint estimation of multiple temperatures can be achieved even in the absence of correlations among the ancillas, surpassing the corresponding thermal Fisher information limits. Exploiting correlations within the ancillary system yields additional enhancement of Fisher information. Finally, we identify the dimensionality of the ancillary systems as a key factor governing the efficiency of multiparameter temperature estimation.

\end{abstract}

\maketitle

\section{Introduction}
Estimating temperatures with precision has widespread applications ranging from assessing how hot the Earth is getting to measuring milli-Kelvin temperatures in a cold-atom Lab. Improving from the thick mercury thermometer lying around in every household, we have now ascended to using the energy levels of an atom to guess the temperature. In this regard, a very useful addition is quantum thermometry \cite{mehboudi2019thermometry, thomas2010, matteo2011, matteo2012, hofer2017, luis2015, Ivanov_2019, rubio2021, DePasquale2018, pati2020, karen2021, alves2022, luiz2022}, where a small probe is sent out to interact and collect information about a bath, and then, using techniques of quantum estimation theory, the temperature of the bath is ascertained. 

The problems of quantum thermometry can be formulated within the broader framework of parameter estimation theory \cite{helstrom1969quantum,paris2004quantum, montenegro2025quantum}, highlighting how temperature can be treated as an estimable quantity governed by the fundamental limits of quantum measurement. The primary prescription of the temperature estimation of a single bath is to thermalize a probe, having a specific heat $C$, to the environment with temperature $T$, and then measuring that probe provides the information about the temperature, which is fundamentally limited by $(\Delta T/T)^2 \ge k_{B}/C$. Over the last decade, various theoretical and experimental techniques have been proposed in the field of quantum thermometry. Apart from thermalizing the probe system, the collisional model \cite{CICCARELLO20221, fiusa2025} has also been used to estimate the temperature of the environment, both in discrete and continuous variable models \cite{seah2019collisional, gabriel2024, mendonça2025}. Starting from single quantum dots to NV centres in nanodiamonds have been designed as precise fluorescent thermometers \citep{neumann2013high,yang2011quantum,kucsko2013nanometre,toyli2013fluorescence,seilmeier2014optical,haupt2014single}. Moreover, quantum harmonic oscillators \cite{brunelli2012qubit}, and Bose-Einstein condensate \cite{sabin2014impurities} have also been used to design quantum thermometers. 

Interestingly, despite substantial recent progress in multiparameter quantum estimation theory~\cite{goldberg2021intrinsic, kaubruegger2023optimal, szczykulska2016multi,ragy2016, Demkowicz-Dobrzański_2020, Gorecki2020optimalprobeserror, liu2020, francesco2022, Candeloro_2024,  Bressanini_2024, mukhopadhyay2025, Mihailescu_2024, gutiérrez2025}, existing approaches have focused almost exclusively on the estimation of a single temperature. Perhaps surprisingly, the development of protocols capable of estimating multiple temperatures remains largely unexplored. Here, we address this gap by developing a protocol for estimating multiple temperatures within a collisional model framework, leveraging tools from multiparameter phase estimation, where $\boldsymbol{\theta} = (\theta_1, \ldots, \theta_N)^{T}$ is encoded in a quantum state $\rho(\boldsymbol{\theta})$. Suitable measurement reveals that the lower bound in measuring the uncertainty is given by the quantum Cram\'er Rao bound (QCRB), i.e., $\mathrm{Cov}(\theta_1, \ldots, \theta_N)\ge \mathcal{F}_Q^{-1}$, where $\mathrm{Cov}(\theta_1, \ldots ,\theta_N)$ is the covariance matrix and $\mathcal{F}_{Q}$ is the Quantum Fisher information matrix (QFIM).

In this paper, we develop a protocol based on a collisional model to estimate the temperatures of multiple thermal reservoirs maintained at different values. We begin by establishing a necessary and sufficient condition to avoid singularities in the estimation of minimum covariance between two temperatures encoded in an ancillary qubit. Building on this result, we introduce into our protocol a controlled unitary rotation which is applied between successive collisions of the qubit ancillas with the probes. We quantify the resulting enhancement of the Fisher information compared to standard temperature-estimation schemes. To assess the performance of the protocol, we define two figures of merit tailored to this multi-parameter estimation scenario. Our analysis reveals that even without preserving quantum correlations among the ancillas, our protocol offers a significant advantage compared to the thermal Fisher informations of the baths. Moreover, by exploiting the quantum correlations of the ancillas, the advantage can be further enhanced. Finally, we address the situation when the number of estimable temperatures is more than two, and in this scenario we discuss the necessity of higher-dimensional ancillas in our proposed framework.\\
The paper is organized as follows. We start by describing our protocol for estimating multiple temperatures using repeated probe-ancilla collisions in Sec .~\ref {framework}. In Sec .~\ref {single_run}, we derive the necessary and sufficient condition for the singularity of the QFIM under certain conditions, and discuss the performance of the protocol for a single ancilla. Sec .~\ref {multiple_run} deals with the effect of uncorrelated and correlated ancillas to acquire information about the thermal baths. The role and necessity of the higher-dimensional ancilla are discussed in Sec .~\ref {higher_dim}. We propose a possible experimental platform in Sec .~\ref {experiment}. Finally, we draw our conclusions in Sec .~\ref {conclusion}.


\section{Joint temperature estimation protocol for multiple baths}
\label{framework}

Let us consider a collection of $N$ non-interacting baths or thermal environments denoted by $B_{1}$, $B_{2}$, \ldots, $B_{N}$ with temperatures $T_{1}$, $T_{2}$, \ldots, $T_{N}$, respectively. Let $\{S_{1}, \ldots,S_{N}\}$  be $N$ probe systems with local Hamiltonians $\{H_{S_1}, \ldots, H_{S_N}\}$ each kept at thermal equilibrium with its corresponding bath. The thermalized probes carry information about their respective bath temperatures $T_{i}$'s $(i = 1, \ldots, N)$. 
To estimate these temperatures, we propose a collisional model framework where these probes are made to interact or ``collide" according to a specific protocol with a stream of independent and identically prepared ancilla systems $\{A_{k}\}$ with $k = 1,2, \ldots n$. Following this, appropriate measurements are performed on the ancillas, which reveal the information about the temperatures. 

We consider a simple collision protocol between the ancillas and the baths. At first $A_1$ collides with $S_{1}$. During this interaction $A_1$ acquired information about the first bath. $A_1$ is then rotated through a unitary $U_{rot}(\theta,\hat r)$ for an angle $\theta$ around an axis $\hat r$. Simultaneously, the probe $S_{1}$ undergoes partial thermalization through a thermal map $\Lambda_{S_{1}}$ for a time $t_{S_{1}B_{1}}$ due to its coupling to the bath $B_{1}$ . This rotated $A_1$ then interacts with $S_{2}$ which is followed by another rotation $U_{rot}(\theta,\hat r)$. This process of successive collision and rotation continues till the last collision with $S_{N}$, and then the ancilla is collected for measurement. This whole process is then repeated, starting with ancilla $A_2$, followed by ancilla $A_3$, and so on (see Fig. \ref{schematic}). At last, separate or joint measurements are performed on the ancillas to estimate the encoded parameters, i.e., temperatures of the multiple baths.

Initially, the system probes are thermalized with the respective baths, and the initial state of the $i$-th system is $\rho_{S_{i}}^{th}$. The initial state of the $k$-th ancilla is $\rho_{A_{k}}$. The system-ancilla interaction is governed by the unitary evolution $U^{i}_{S_{i}A_{k}} = \exp(-i H_{int}^{S_{i}A_{k}}t^{i}_{S_{i}A_{k}})$ for a time interval $t^{i}_{S_{i}A_{k}}$ $(i = 1, \ldots, N)$. Here, $H_{int}^{S_{i}A_{k}}$ is the interaction Hamiltonian between the $i$-th system and $k$-th ancilla. The protocol follows three steps: first, the ancilla interacts with the first probe via $U^{1}_{S_{1}A_{k}}$, and the resultant ancilla state reads as,
\begin{align}
    \tilde{\rho}_{A_k}(1)&=\mathcal{E}_{S_1A_k}\big[\rho_{A_{k}}\big]\nonumber\\
    &= \mathrm{Tr}_{S_{1}} [(U^{1}_{S_{1}A_{k}}) \big(\Lambda_{S_1}^{k-1}[\rho_{S_{1}}^{th}]\otimes \rho_{A_{k}}\big) (U^{1}_{S_{1}A_{k}})^{\dagger}]
\end{align}

Here, $\Lambda_{S_{1}}^{k-1} [\rho_{S_{1}}^{th}]$ denotes the state of the first probe available to the $k$-th ancilla. Mathematically, $\Lambda_{S_{i}}$ is defined as $\Lambda_{S_{i}}[\rho_{S_{i}}] = \gamma (\bar{n} + 1) \mathcal{K}[\sigma_{-}^{S_{i}}] + \gamma \bar{n} \mathcal{K}[\sigma_{+}^{S_{i}}]$ in the interaction picture, with $\mathcal{K} (O) = O \rho_{S_{i}} O^{\dagger} - 1/2 \{ O O^{\dagger}, \rho_{S_{i}} \} $ and $\sigma_{\pm}^{S_{i}}$ being the raising and lowering operators acting on the probe $S_{i}$. $\gamma$ is the time-independent coupling constant, and $\bar{n}$ is the thermal occupation number. This is followed by the rotation of the ancillary system, given by
\begin{align*}
     \tilde{\tilde{\rho}}_{A_k}(1)=\mathcal{U}_{rot}\big[\tilde{\rho}_{A_k}(1)\big]= U_{rot}(\theta,\hat{r}) \,\, \tilde{\rho}_{A_k}(1)\,\, U_{rot}(\theta,\hat{r})^{\dagger}.
\end{align*}
\begin{figure}
\includegraphics[width=\linewidth]{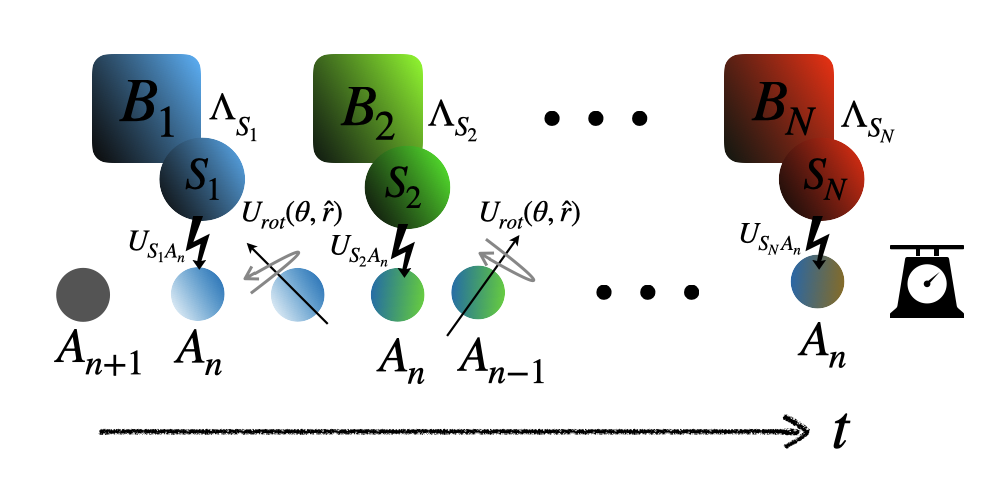}
\caption{Schematic representation of the multiparameter temperature estimation protocol based on the collisional model.
    The information about the first bath, $B_{1}$, is acquired by the $n$-th ancilla through its interaction with the probe system $S_{1}$. 
    A controlled unitary rotation is subsequently applied to the ancilla $A_{n}$, after which it interacts with the next probe to extract information about the second bath, $B_{2}$. This sequence is iteratively repeated for all $N$ thermal reservoirs. }
\label{schematic}
\end{figure}
At last, the resultant ancilla interacts with the next bath $S_{2}$, and the process continues. The resultant ancilla after $N$ successive collisions and rotations looks like
\begin{align}
    \rho^{f}_{A_k}(N)= \mathcal{U}_{rot}\circ \mathcal{E}_{S_NA_k}\circ ...\circ \mathcal{U}_{rot}\circ \mathcal{E}_{S_1A_k}\big[\rho_{A_{k}}\big].
\end{align}

Here, $\rho^{f}_{A_k}(N)$ denotes the final state of the $k$-th ancilla. The protocol now continues for $n$ ancillas, and the final joint state of the ancillas, $\rho_{A_{1}\ldots A_{n}}$, is measured via a positive operator valued measurement (POVM) $\{\Pi_{j}\}_{j}$ ($j$ is the possible outcome), which provides the Fisher information matrix $F(N)$. Maximizing over all such POVMs provides the QFIM $\mathcal{F}^{(N)}_{Q}$, whose $(i,j)$ element is defined by, $(\mathcal{F}_Q)_{ij} = \frac{1}{2}\mathrm{Tr}[\rho^{f}_{A_{1} \ldots A_{n}} (L_iL_j+L_jL_i)]$. Note that $L_{\mu}$ is the symmetric logarithmic derivative (SLD) corresponding to the temperature $T_{\mu}$ given by $\tfrac{\partial \rho^f_{A_1, \ldots, A_n}}{\partial T_{\mu}}=\tfrac{1}{2}[L_{\mu}\rho^f_{A_1, \ldots, A_n}+\rho^f_{A_1,.., A_n}L_{\mu}]$.\\

\textit{Performance indicator: } Here we consider two figures of merit to compare our protocol with the scenario when the temperature of the individual baths is estimated separately using their respective thermal Fisher informations $\mathcal{F}_{th}^{i}$. We compare the efficiency of our protocol by constructing a $N \times N$ thermal Fisher information matrix $\mathcal{F}_{th}$ for $N$ baths.

Suppose a probe is in thermal contact with a bath having inverse temperature $\beta = 1/k_{B}T$, where $k_{B}$ is the Boltzmann constant. The completely thermalized state of the system is $\rho_{S_{i}}^{th} = \exp(-\beta H_{S_{i}})/{\exp(-\beta H_{S_{i}})}$, where $H_{S_i}$ is the local Hamiltonian of the probe. Then, the thermal Fisher information acquired from the system is given by
\begin{equation}
    F_{th}^{i} = \frac{\langle H_{S_i}^{2} \rangle - \langle H_{S_i} \rangle^2}{k_{B}^2T^{4}}.
    \label{thermal_fisher}
\end{equation}

  where $\langle \bullet \rangle$ denotes the average value of an operator. Using this, we construct the thermal Fisher information matrix $\mathcal{F}_{th}$ as $\mathcal{F}_{th} = diag(F_{th}^{1}, F_{th}^{2}, \ldots F_{th}^{N})$. $F_{th}^{i}$ is the thermal Fisher information for the $i$-th probe. In this scenario, the total input state is $\rho_{th} = \otimes_{i=1}^{N}\rho_{S_i}^{th}$ and $\rho^{th}_{S_i}$ only contains the information about the $i$-th bath. Note that the construction of $\mathcal{F}_{th}$ has been designed to serve as a benchmark for our proposed protocol.  We compare these two Fisher information matrices $\mathcal{F}_Q$ and $\mathcal{F}_{th}$ and \textit{total average information} as, 
    \begin{equation}
        \eta_{\mathrm{joint}} = \frac{\mathrm{Tr}(\mathcal{F}_{Q})}{\mathrm{Tr}(\mathcal{F}_{\mathrm{th}})}.
    \end{equation}

Here, $\eta_{\mathrm{joint}} < 1 (>1)$ indicates that the joint information contained in $\mathcal{F}_Q$ is smaller (larger) than $\mathcal{F}_{th}$. Note that $\eta_{joint}$ does not consider the effect of the correlation between different parameters. However, we also define another performance indicator, \textit{information accuracy,} 
\begin{equation}
\eta_{\mathrm{acc}} = \log \bigg[\frac{\det(\mathcal{F}_Q)}{\det(\mathcal{F}_{\mathrm{th}})}\bigg],
\end{equation}
that captures the effect of correlation among the parameters, and also ensures the minimum error condition. The positive value of $\eta_{\mathrm{acc}}$ indicates that the lower bound of error in the joint estimation protocol is less than $\mathcal{F}_{th}$. On the other hand, a negative value implies that the estimation protocol has a large lower bound in the error estimation. In the rest of the manuscript, we will use these two indicators to compare the effectiveness of our proposed protocol.




\section{Single-run temperature estimation}
\label{single_run}

We begin by considering the simplest model, consisting of two probes, each of which is fully thermalized with its respective bath. A single ancilla interacts with these probes sequentially to acquire information about the bath parameters. In multiparameter estimation, one of the central difficulties arises from the singularity of the QFIM ($\mathcal{F}_{Q}$), i.e., $\mathcal{F}_{Q}^{-1}$ does not exist. This leads to an unbounded lower bound on the estimation uncertainty. This divergence reflects an intrinsic inefficiency in the design of the protocol. Such singularities can occur for various reasons~\cite{yang2025}, one of which is that the parameters of interest may become functionally dependent. Consequently, an efficient multiparameter estimation scheme must explicitly account for this issue.
In our proposed protocol, a controlled unitary rotation introduced between the two encoding stages eliminates these parameter dependencies. In this regard, we now proceed to prove the following theorem:

\begin{figure}
\includegraphics[width=\linewidth]{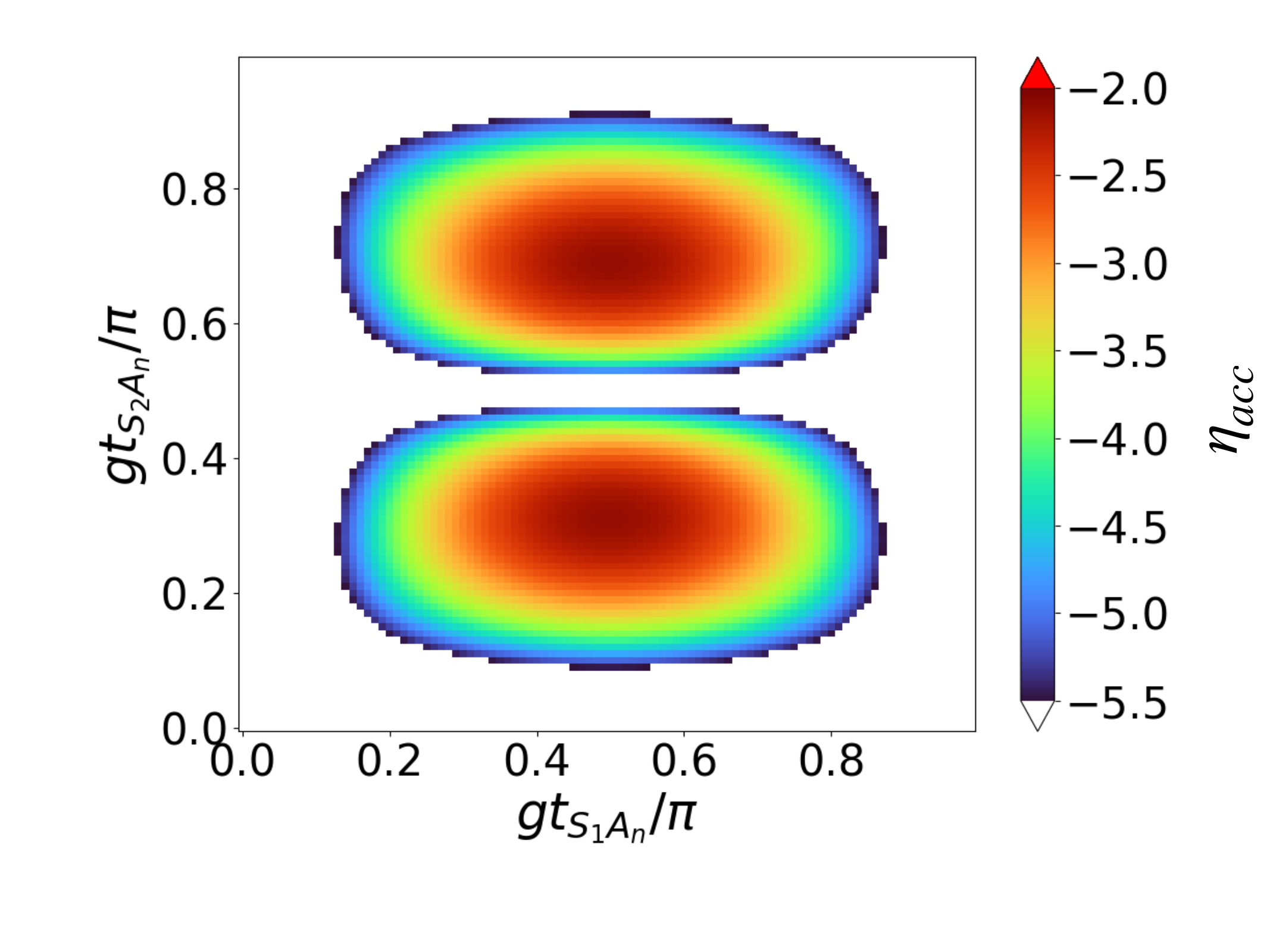}

\caption{The variation of $\eta_{\mathrm{acc}}$ is illustrated for the case where a single ancilla sequentially interacts with the first and second probes, with interaction strengths $gt_{S_{1}A_{n}}/\pi$ and $gt_{S_{2}A_{n}}/\pi$, respectively. Here, $K_{B}T_{B_{1}}/\hbar\omega = 2$, and $K_{B}T_{B_{2}}/\hbar\omega = 1$.
 }
 \label{2baths_one_anc}
\end{figure}

\begin{theorem}\label{theorem1}
    For any bi-parameterized qubit state, $\eta(\xi_1,\xi_2)$ the QFIM $\mathcal{F}_Q(\eta)$ is non-invertible, i.e., $\text{det}\big(\mathcal{F}_Q(\eta)\big) = 0$  if and only if
    \begin{equation}\label{theorem_condition}
        \frac{\partial \eta}{\partial \xi_1} = c \frac{\partial \eta}{\partial \xi_2}
    \end{equation}
    for some real number $c$.
\end{theorem}
\begin{proof}
    It can be easily seen that the theorem holds true for the trivial case $\frac{\partial\eta}{\partial\xi_1}=0$ or $\frac{\partial\eta}{\partial\xi_2}=0$ by choosing $c=0$. We therefore consider the non-trivial scenario where neither of the derivatives is zero. It ensures that the SLD operators $L_i$'s are non-zero and $\mathcal{F}_{Q}$ is either rank $1$ or $2$. Using the spectral decomposition $\eta=\alpha_0\ket{\alpha_0}\bra{\alpha_0}+\alpha_1\ket{\alpha_1}\bra{\alpha_1}$ the QFIM can be written as
\begin{align}
    \mathcal{F}_{Q}(\eta)= \tfrac{1}{2}\alpha_0 \big(M_0+M_0^*\big)+\tfrac{1}{2}\alpha_1 \big(M_1+M_1^*\big)
\end{align}
where $M^*$ denotes complex conjugate of a matrix $M$, and, $M_0$ and $M_1$ are positive matrices given by $M_0=W_0W_0^{\dagger}$ and $M_1=W_1W_1^{\dagger}$, where
\begin{align}
    W_0=\begin{pmatrix}
        \bra{\alpha_0}L_1 \\ \bra{\alpha_0}L_2
    \end{pmatrix}\,\,\,;\,\,\,
    W_1=\begin{pmatrix}
        \bra{\alpha_1}L_1 \\ \bra{\alpha_1}L_2
    \end{pmatrix}
\end{align}
Note that the SLD operators can be expressed as
\begin{align}\label{sld_defn}
    L_i= \sum_{j,k=0,1}\frac{2\bra{\alpha_j}\frac{\partial \eta}{\partial T_i}\ket{\alpha_k}}{\alpha_j+\alpha_k} \ket{\alpha_j}\bra{\alpha_k}
\end{align}
Also note that the rank of any $M^*$ is always equal to that of $M$. It can be easily seen from Eq. (\ref{sld_defn}) that condition (\ref{theorem_condition}) is equivalent to
\begin{align}
    \bra{\alpha_0}L_1\ket{\alpha_0} &= c\,\bra{\alpha_0}L_2\ket{\alpha_0} \label{eq_condition_1}\\
    \bra{\alpha_0}L_1\ket{\alpha_1} &= c\,\bra{\alpha_0}L_2\ket{\alpha_1} \label{eq_condition_2}\\
    \bra{\alpha_1}L_1\ket{\alpha_1} &= c\,\bra{\alpha_1}L_2\ket{\alpha_1} \label{eq_condition_3}
\end{align}
Therefore, if condition (\ref{theorem_condition}) is true then from Eqs. (\ref{eq_condition_1}), (\ref{eq_condition_2}) and (\ref{eq_condition_3})  we can conclude that $\bra{\alpha_0}L_1=c\bra{\alpha_0}L_2$ and $\bra{\alpha_1}L_1=c\bra{\alpha_1}L_2$. This implies
\begin{align}
    W_0=\begin{pmatrix}
        c\bra{\alpha_0}L_2 \\ \bra{\alpha_0}L_2
    \end{pmatrix}\,\,;\,\,W_1=\begin{pmatrix}
        c\bra{\alpha_1}L_2 \\ \bra{\alpha_1}L_2
    \end{pmatrix}
\end{align}
As a result, we get
\begin{equation*}
    \mathcal{F}_Q(\eta)=\big(\alpha_0\bra{\alpha_0}L_2^2\ket{\alpha_0}+\alpha_1\bra{\alpha_1}L_2^2\ket{\alpha_1}\big)\begin{bmatrix}
        c \\ 1
    \end{bmatrix}\begin{bmatrix}
        c & 1
    \end{bmatrix}
\end{equation*}
which clearly implies $\mathcal{F}_Q(\eta)$ is rank $1$ and as a result $\text{det}(\mathcal{F}_Q(\eta)) = 0$. This proves the first part of the theorem.

We now show if condition (\ref{theorem_condition}) is violated, $\mathcal{F}_{Q}$ is rank $2$ and consequently $\text{det}(\mathcal{F}_Q(\eta))\neq 0$. From the equivalence of condition (\ref{theorem_condition}) and  Eqns. (\ref{eq_condition_1}), (\ref{eq_condition_2}) and (\ref{eq_condition_3}) we can conclude that if condition (\ref{theorem_condition}) is violated at least one of the following inequalities are true

\begin{align}
    \frac{\bra{\alpha_0}L_1\ket{\alpha_0}}{\bra{\alpha_0}L_2\ket{\alpha_0}} &\neq \frac{\bra{\alpha_0}L_1\ket{\alpha_1}}{\bra{\alpha_0}L_2\ket{\alpha_1}}\label{violation_0} \\
    \frac{\bra{\alpha_1}L_1\ket{\alpha_0}}{\bra{\alpha_1}L_2\ket{\alpha_0}} &\neq \frac{\bra{\alpha_1}L_1\ket{\alpha_1}}{\bra{\alpha_1}L_2\ket{\alpha_1}}\label{violation_1}
\end{align}
Inequality (\ref{violation_0}) implies $W_0$ is rank $2$, and, inequality (\ref{violation_1}) implies $W_1$ is rank $2$. Thus at least one of $W_0$ or $W_1$ is rank $2$, which in turn means $\mathcal{F}_{Q}$ is rank $2$.
\end{proof}

Suppose the local Hamiltonian of the two baths and the ancilla are $H_{S_{i}} = \omega_{S_{i}} \sigma_{z}/2$ ($i = 1,2$) and, $H_{A_{1}} = \omega_{A_{1}} \sigma_{z}/2$ respectively. The interaction Hamiltonian between the system and ancilla is given by, $H_{S_{i}A_{1}}=g_{i}(\sigma_{+}\otimes \sigma_{-}+\sigma_{-}\otimes \sigma_{+})$. Let us consider the initial state of the ancilla to be $\rho_{A_{1}}$, and the thermalized probe is in a state $\rho^{th}_{S_{i}} =\sum_{k}\lambda_k(T_{i})\ket{k}\bra{k}$ which are thermalized with the corresponding bath having temperature $T_{i}$. The action of the collision is given by
\begin{align*}
    \mathcal{E}_{T_{1}}[\rho_{A_{1}}]=Tr_{S_{1}}\big[U_{S_{1}A_{1}}(\rho_{S_{1}}^{th}(T_{1})\otimes\rho_{A_{1}})U_{S_{1}A_{1}}^{\dagger}\big]
\end{align*}


Now for two consecutive collisions of the ancilla with the first and the second bath, we have $\hat{\hat{\mathcal{E}}}_{f} =  \hat{\hat{\mathcal{E}}}_{T_2}\cdot\hat{\hat{\mathcal{E}}}_{T_1}$. Note that $\hat{\hat{\bullet}}$ indicates the vectorized form of the operators. For the sake of simplicity, we take the interaction strengths to be identical, i.e., $g_{1} = g_{2}$ and $t_{S_{1}A_{1}} = t_{S_{2}A_{1}} = t$. Now, if the initial ancilla state is $\rho_{A_{1}}=\ket{1}\bra{1}$, we get $\hat{\hat{\mathcal{E}_{f}[\rho_{A_{1}}]}}=\hat{\hat{\Lambda}}_f\cdot \hat{\hat{\rho_{A_{1}}}}$ (for detailed calculation, see Appendix. \ref{appendix_u_rot}), i.e., 
\begin{align}
    \rho^f_{A_{1}}=\mathcal{E}_f[\rho_{A_{1}}] = \begin{bmatrix}
        v(T_1,T_2) & 0 \\
        0 & 1-v(T_1,T_2)
    \end{bmatrix},
\end{align}
where $v(T_1,T_2) = \sin^2 gt \Big[p+q \cos^2 gt \Big]$ with $p=\lambda_0(T_1)$ for $\hat{\hat{\mathcal{E}}}_{T_1}$ and $q=\lambda_0(T_2)$ for $\hat{\hat{\mathcal{E}}}_{T_2}$.

It can be easily seen that $\tfrac{\partial \rho^{f}_{A_{1}}}{\partial T_1}=\big[\big(\tfrac{\partial v}{\partial T_1}\big)/\big(\tfrac{\partial v}{\partial T_2}\big)\big]\,\tfrac{\partial \rho^{f}_{A_{1}}}{\partial T_2}$. Thus, using Theorem \ref{theorem1}, we must have $\text{det}(\mathcal{F}_{Q}(\rho^{f}_{A_{1}}))=0$. Moreover, the exact form of the QFIM in this case is given by
\begin{align}
    \mathcal{F}_{Q}(\rho^f_{A_1})=\frac{1}{v(1-v)}\begin{bmatrix}
        (\partial_1 v)^2 & (\partial_1 v)(\partial_2 v) \\
        (\partial_2 v)(\partial_1 v) & (\partial_2 v)^2
    \end{bmatrix} 
\end{align}
As can be easily seen here, $\text{det}(\mathcal{F}_{Q})=0$, which implies that the protocol automatically makes the two parameters dependent. In order to avoid this issue, we introduce a rotation $U_{rot}(\pi/4,\hat{x}) = \mathrm{exp}(-i\pi\sigma_{x}/4)$ between two successive collisions as shown in Fig. \ref{schematic}. 
Thus, $\hat{\hat{\mathcal{E}}}_f$ is given by
\begin{align}
    &\hat{\hat{\mathcal{E}}}_{f} =  \hat{\hat{\mathcal{E}}}_{T_2}\cdot\hat{\hat{\mathcal{U}}}_{rot}\cdot\hat{\hat{\mathcal{E}}}_{T_1} 
\end{align}

One can easily see that $\rho^{f}_{A_{1}}=\mathcal{E}_f[\rho_{A_{1}}]=\hat{\hat{\mathcal{E}}}_f\cdot\hat{\hat{\rho_{A_{1}}}}$ is in the form
\begin{equation}
    \rho^{f}_{A_{1}}=\begin{bmatrix}
    \mu (q) & -i\chi(p) \\
    i\chi(p) & 1-\mu(q)
    \end{bmatrix}
\end{equation}
where $\mu(q) = \tfrac{1}{4}\!\big[(1 + 2q) + (1 - 2q)\cos 2g\big]$ and $\chi(p) = \tfrac{1}{2} \big[(1 - p) + p\cos 2g\big]\cos g $. 
Now, $\tfrac{\partial}{\partial T_1}\big(\rho^{f}_{A_{1}}\big) =\dot{\chi}\sigma_y$ and $ \tfrac{\partial}{\partial T_2}\big(\rho^{f}_{A_{1}}\big) =\dot{\mu}\sigma_z$,
where $\dot{\mu}=d\mu/ d T_2$ and $\dot{\chi}=d\chi/ d T_1$. Using Theorem \ref{theorem1} we can conclude $\text{det}(\mathcal{F}_{Q})\neq 0$.


Since in this case $[L_{1}, L_{2}] \neq 0$ (see Appendix \ref{appendix_u_rot}), the two temperatures cannot, in general, be estimated simultaneously. Nevertheless, within the framework of our protocol, one can appropriately tune the free parameters and reach an estimation strategy that is closest to the simultaneous measurement scenario. In particular, the achievable Fisher information for both baths can be made comparable to their respective thermal Fisher information.
For a single-collision protocol, one typically expects the information-accuracy $\eta_{\mathrm{acc}}$ to be negative, indicating that the lower bound on temperature estimation in our scheme is higher than the bound set by $\mathcal{F}_{\mathrm{th}} = \mathrm{diag}(0.015,0.197)$ (from Eq. \ref{thermal_fisher}). However, as shown in Fig.~\ref{2baths_one_anc}, there exist certain collision strengths even within a single-run scenario for which $\eta_{\mathrm{acc}}$ gets closer to zero, although still being negative. We observe that when $g t_{S_{1}A_{n}}/\pi = 0.5$, the protocol attains the maximum value of $\eta_{\mathrm{acc}}$, or in other words, the minimum estimation uncertainty predicted by the QCRB. Note that this value makes the interaction unitary between the ancilla and the first bath to be a SWAP operation with an imaginary phase. On the other hand, the total information jointly accessible about the two parameters does not depend on the specific nature of the first collision. However, to achieve the maximum value of $\eta_{\mathrm{joint}}$, the second interaction must effectively implement a phased-SWAP operation. Typically, to obtain the maximum joint information about the parameters, one should fix the interaction strength so that a phased-SWAP operation is performed on both probes. However,  the estimation uncertainty becomes considerably larger at this interaction strength, indicating a clear trade-off between $\eta_{\mathrm{joint}}$ and $\eta_{\mathrm{acc}}$. Taking this into account, we exclude this regime from our analysis. In the remainder of the manuscript, we fix the interaction strength between the ancilla and the first bath such that this phased-SWAP operation is reproduced.

\begin{figure}
\includegraphics[width=\linewidth]{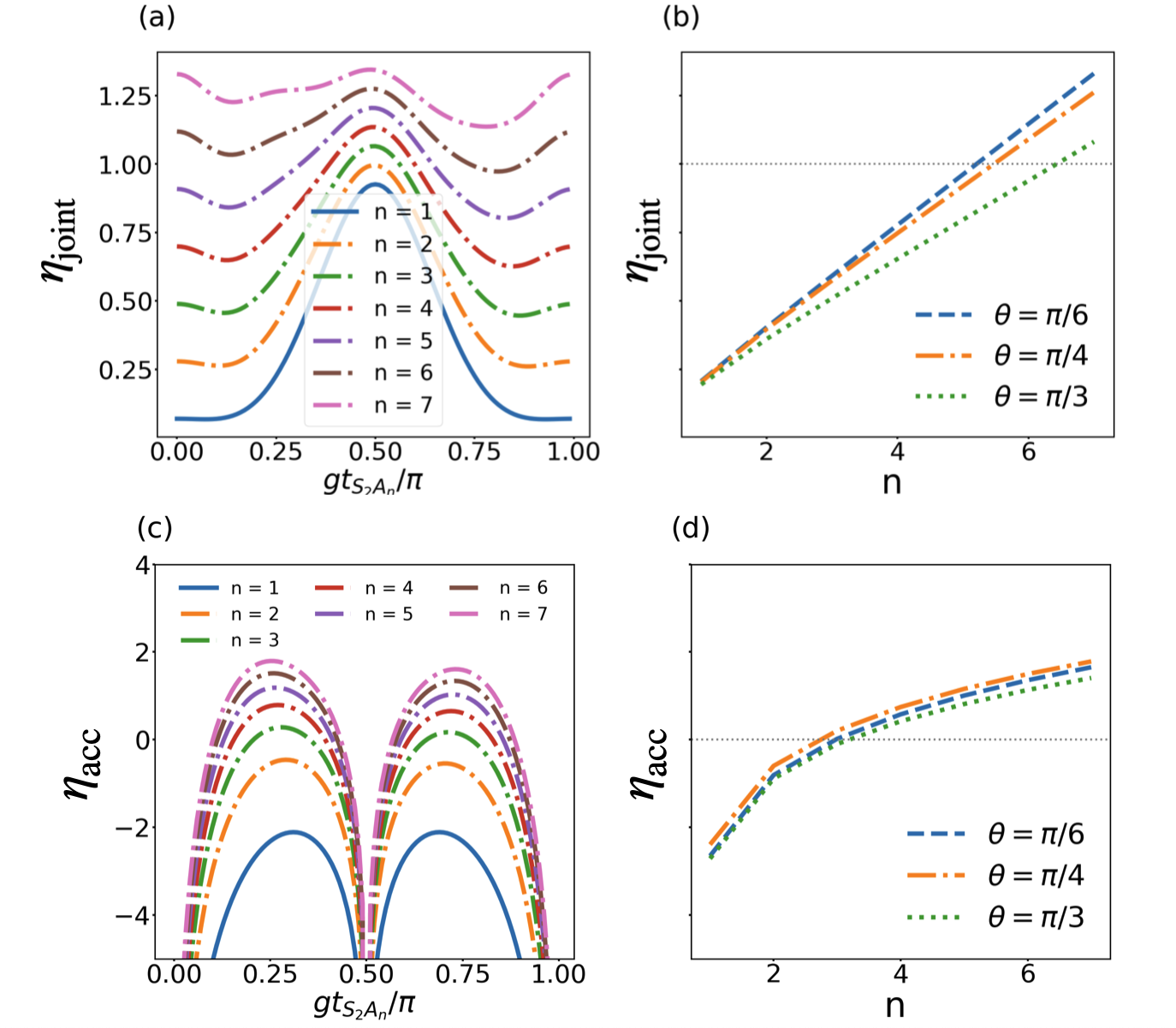}

\caption{The trend of $\eta_{\mathrm{joint}}$ and $\eta_{\mathrm{acc}}$ with respect to coupling strength for different numbers of uncorrelated ancillas $n$ ($(a)$ and $(c)$). Here $gt_{S_{1}A_{n}}/\pi = 0.5$ and $\gamma t_{S_{i}B_{i}} = 0.5$ for $i = 1,2$. $(b)$ and $(d)$ denotes the scaling of $\eta_{\mathrm{joint}}$, and $\eta_{\mathrm{acc}}$ for different angles of rotation around the $x$-axis. The dotted horizontal lines in $(b)$ and $(d)$ indicate $\eta_{\mathrm{joint}} = 1$ and $\eta_{\mathrm{acc}} = 0$ respectively i.e., $\mathcal{F}_{Q} = \mathcal{F}_{th}$.
All other details are the same as Fig. \ref{2baths_one_anc}.} 
\label{2baths_multi_anc}
\end{figure}
\section{Information acquired through multiple collisions} 
\label{multiple_run}
Previous analysis based on a single-collision setup shows that, for specific values of the interaction strength, it becomes possible to extract information about the parameters of interest within well-defined error bounds. Let us now discuss whether it is possible to gain more information by exploiting multiple ancillary systems.  Here we consider two specific scenarios: one, performing a joint measurement on the resultant product ancilla state of the $n$ qubits (where $n$ is the total number of ancilla) following the protocols as described in Sec. \ref{framework}. Secondly, we analyze the situation when the ancillary systems get correlated due to the successive interaction with the same probe systems.

\textit{Effect of uncorrelated multiple ancilla:} A systematic analysis reveals that employing multiple ancillary systems provides a significant advantage within the proposed protocol. Even in the absence of correlations among the ancillas, performing a global measurement on their joint state enables the extraction of more information than in the single-collision scenario. As illustrated in Fig.~\ref{2baths_multi_anc}(a), when the number of ancillas exceeds three, the performance of our protocol surpasses the joint information contained in $\mathcal{F}_{\mathrm{th}}$, i.e., $\eta_{\mathrm{joint}} > 1$. In this regime, a finite error bound is also achievable, as shown in Fig.~\ref{2baths_multi_anc}(c). Intuitively, after each interaction with an ancilla, the individual probe systems are allowed to partially thermalize with their respective baths. As a result, when the next ancilla interacts with the probe, the probe has already acquired additional information from its corresponding bath. This repeated interaction of the probe with the bath between successive ancilla collisions provides an advantage over the single-collision model, even in the absence of any correlations between the ancillas.

Moreover, after more than three collisions, the maximum numerical value of $\eta_{\mathrm{acc}}$ becomes positive, indicating that the estimation error in our protocol is lower than the corresponding thermal Fisher-information bound. The scaling of $\eta_{\mathrm{joint}}$ with respect to the number of ancillas, presented in Fig.~\ref{2baths_multi_anc}(b), demonstrates that it increases monotonically; hence, larger numbers of ancillas yield progressively more joint information about the bath temperatures. Nevertheless, Fig.~\ref{2baths_multi_anc}(d) suggests that the error eventually saturates as the number of collisions grows. Despite this saturation, the error bound remains considerably tighter than in the thermal scenario, since $\eta_{\mathrm{acc}} > 0$ even after three consecutive collisions.

It is also worth noting that each ancilla undergoes the same rotation by an angle $\theta$. While $\theta = \pi/6$ yields slightly higher joint information, the error bound is minimized for $\theta = \pi/4$, where $\eta_{\mathrm{acc}}$ reaches its maximum. Consequently, in our subsequent analysis of the influence of correlations between ancillas, we focus exclusively on the case $\theta = \pi/4$.

\textit{Information gain from ancilla correlation:} Up to this point, we have examined the scenario in which the ancillas, despite undergoing repeated interactions with the same probe system, remain uncorrelated under our proposed protocol. We now turn to the case where the ancillas are correlated, in which the final state of the ancillary systems, from which the information about the two baths ($N = 2$) must be extracted, is given by
\begin{align}
\rho_{A_{1}\ldots A_{n}} 
    &= \mathrm{Tr}_{S_{1}S_{2}} \Big[
        \Lambda_{S_{2}} \circ \mathcal{U}^{A_n}_{rot} 
        \circ \mathcal{U}_{S_{2}A_{n}} \circ \cdots \nonumber\\
    &\quad
        \cdots \circ \Lambda_{S_{2}} \circ \mathcal{U}^{A_2}_{rot} 
        \circ \mathcal{U}_{S_{2}A_{1}} \nonumber\\
    &\quad
        \circ \Lambda_{S_{1}} \circ \mathcal{U}^{A_1}_{rot} 
        \circ \mathcal{U}_{S_{1}A_{1}}
        \big( \rho^{S_{1}}_{th} \otimes \rho^{S_{2}}_{th} 
        \otimes \rho_{A_{k}}^{\otimes_{k=1}^n} \big)
        \Big].
\end{align}
Here $\mathcal{U}_{rot}^{A_{i}}$ denotes the rotation is applied in the $A_{i}$-th ancilla. Performing suitable measurements on the joint ancilla state $\rho_{A_{1}\ldots A_{n}}$ reveals that for $n = 2$, the correlated and uncorrelated scenarios provide almost identical values of the joint information and accuracy bound (see Fig. \ref{corr_2_bath}). However, when the number of ancillary systems is increased, the correlations established among them provide a significant advantage over both the uncorrelated case and, naturally, the thermal reference scenario. Moreover,
\begin{figure}[H]
\includegraphics[width= \linewidth]{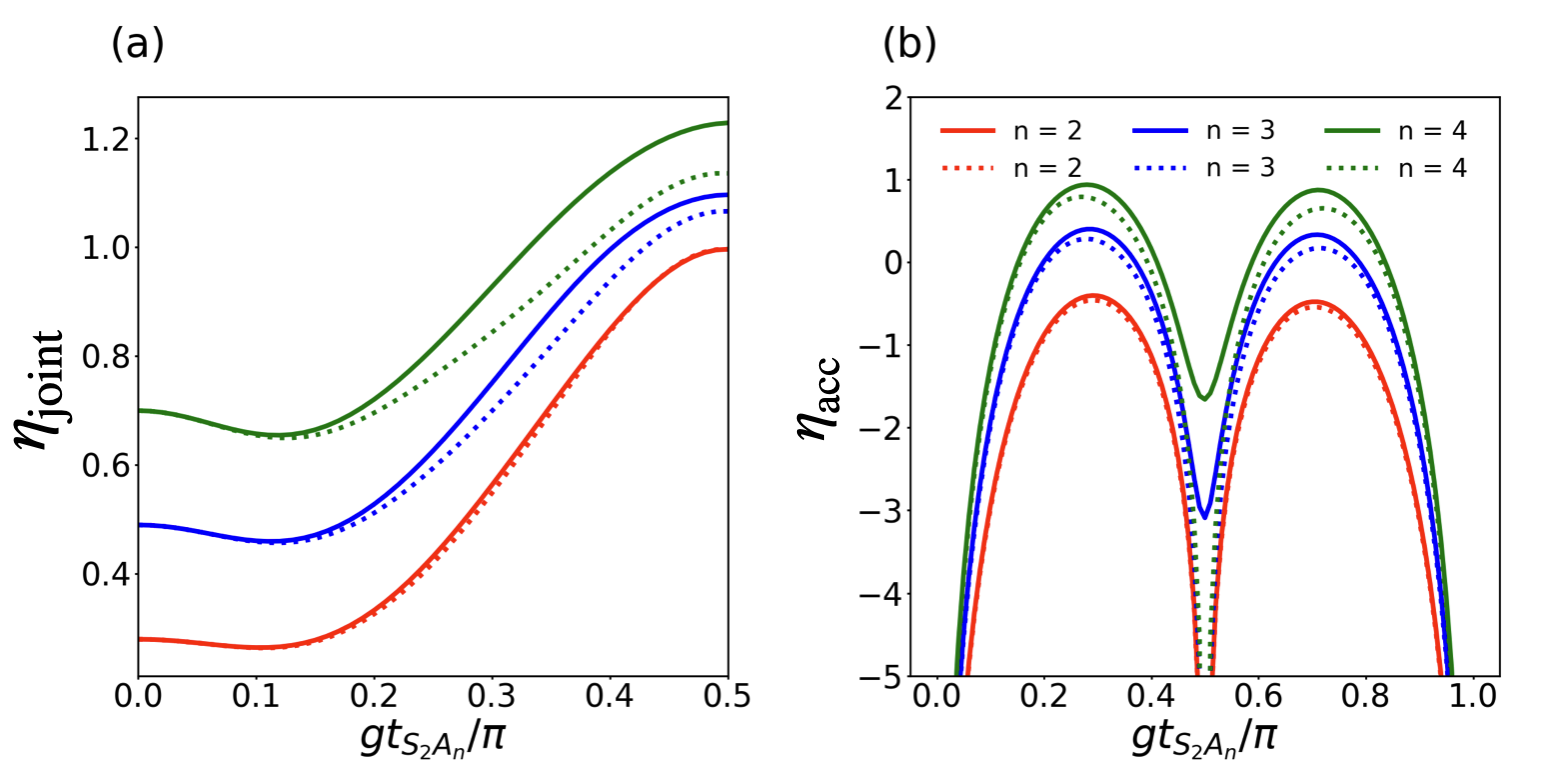}
\caption{ Comparison between correlated (solid lines) and uncorrelated ancilla (dotted lines) for (a) $\eta_{\mathrm{joint}}$ and (b) $\eta_{\mathrm{acc}}$ with respect to the interaction strength. Here, $g t_{S_{1}A_{n}} = 0.5$. All other details are similar to Fig. \ref{2baths_multi_anc}.}
\label{corr_2_bath}
\end{figure}
$\eta_{\mathrm{acc}}$ increases substantially, demonstrating that our proposed protocol is highly efficient in accurately inferring the temperatures of the two baths.

\section{Role of higher-dimensional ancilla in the estimation protocol} 
\label{higher_dim}

In this section, we examine the problem of temperature estimation involving three thermal baths, corresponding to \( N = 3 \). We begin by considering two-dimensional ancillary systems employed to encode information about the baths. Since a qubit possesses three degrees of freedom, it is, in principle, capable of encoding three independent parameters \cite{liu2020}. Nevertheless, it has been shown that under unitary encoding, the simultaneous estimation of multiple parameters becomes infeasible when the number of parameters exceeds the dimension of the probe, as the Fisher information matrix turns singular. Consequently, no finite bound on the estimation precision can be established \cite{Candeloro_2024}. In our proposed protocol, although the encoding process is modelled as a general completely positive trace-preserving (CPTP) map rather than a strictly unitary evolution, we observe a similar limitation. While the qubit can, in principle, carry information about three distinct temperatures, the parameter \(\eta_{\mathrm{acc}}\) becomes significantly negative, indicating that the achievable error bound within this scheme is highly unsatisfactory.

To avoid this hurdle, we consider a collisional model consisting of identically prepared qutrit systems to encode and decode the information of the baths. To make the analysis coherent, we consider an energy-conserving interaction Hamiltonian between the $i$-th qubit probe and the $n$-th qutrit ancilla with the identical interaction strength, given by, $H_{S_{i}A_{n}} = g(\sigma_{+}\mathcal{Q}_{-} + \sigma_{-}\mathcal{Q_{+}})$. Here, $\mathcal{Q}_{\pm} = (S_{x} \pm iS_{y})/2$ with $S_{\nu}$ ($\nu = x,y,$ and $z$) denote the higher-dimensional Pauli operators and $g$ is the interaction strength.\\

\begin{figure}
\includegraphics[width= \linewidth]{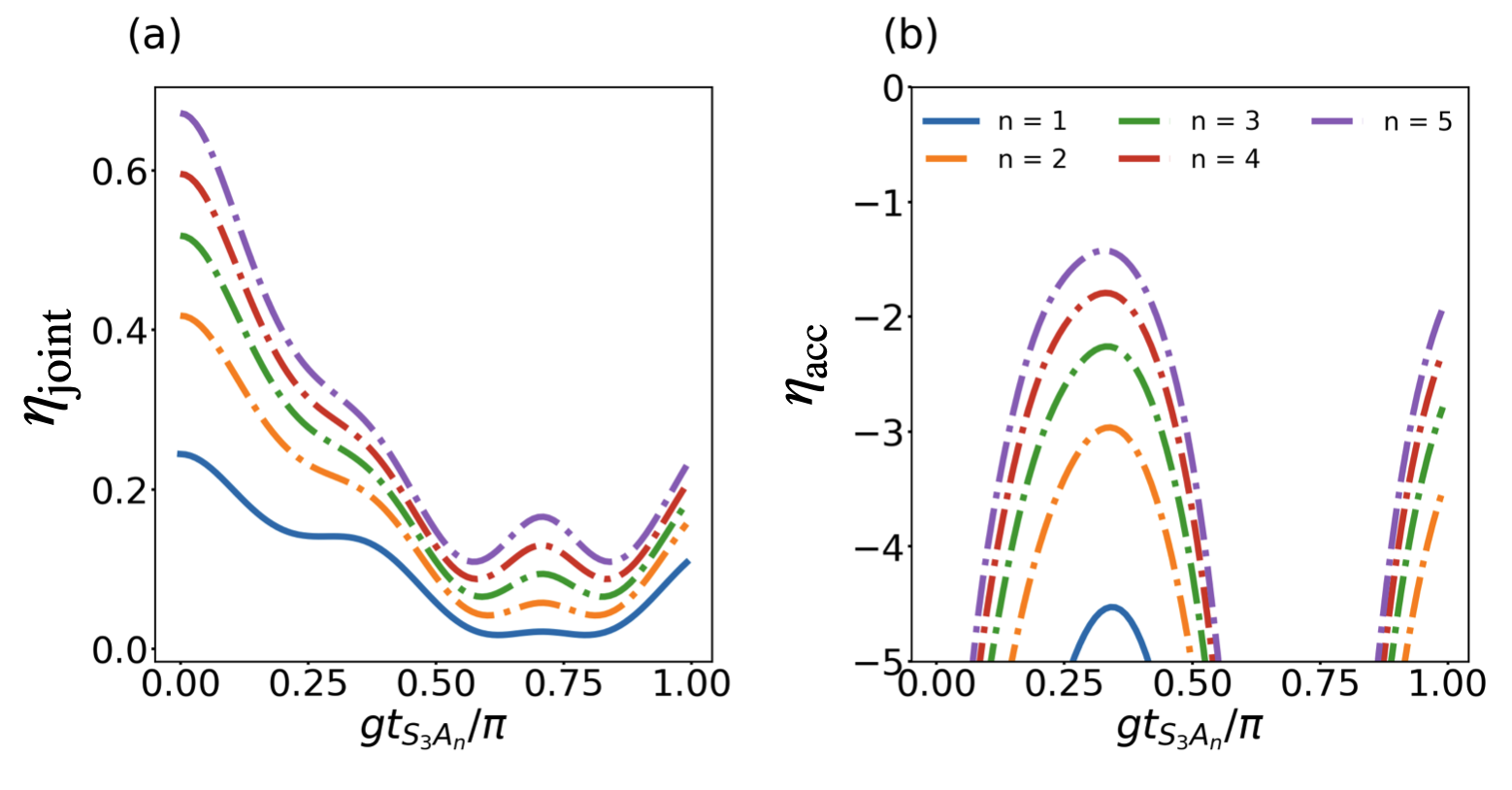}
\caption{ The behaviour of (a) $\eta_{\mathrm{joint}}$ and (b) $\eta_{\mathrm{acc}}$ with respect to the interaction strength of the qutrit with the last probe. Here, $g t_{S_{1}A_{n}} = 0.5$, $g t_{S_{2}A_{n}} = 0.2$ and $K_{B}T_{B_{3}}/\hbar\omega = 3$. All other details are similar to Fig. \ref{2baths_multi_anc}.}
\label{3baths_multi_anc}
\end{figure}

 Fig.~\ref {3baths_multi_anc}(a) indicates that achieving sufficient information about the system---while maintaining a decent lower bound on the estimation error---requires the coupling strength between the third bath and each ancilla to be approximately comparable to the interaction strength with the second probe. Increasing the dimension of the ancillary systems further ensures a finite lower bound on the estimation error, as illustrated in Fig.~\ref{3baths_multi_anc}(b). Additionally, enlarging the number of ancillas enhances the achievable precision, consistent with the behavior observed in the uncorrelated two-bath scenario. Although the case $n = 5$ does not surpass the thermal precision limit given by $\mathcal{F}_{th} = \mathrm{diag}(0.015, 0.197, 0.003)$, it is evident that sufficiently many ancillas will do so. As before, each ancilla undergoes a unitary rotation between successive interactions with the probe systems. These observations collectively demonstrate that our protocol enables accurate joint temperature estimation of $N$ baths, provided that the ancilla dimension is chosen appropriately.

\section{Proposed Experimental Framework}
\label{experiment}


The theoretical architecture considered here—namely, a sequence of two-level systems interacting one at a time with single-mode cavities that, between interactions, relax toward a thermal state—is the standard framework of micromaser and Rydberg cavity-QED physics. In these systems, the intracavity field dissipates into a thermal reservoir, and changes in accordance with a damped harmonic-oscillator master equation while the individual Rydberg atoms traverse a high-Q superconducting cavity undergoing coherent Jaynes–Cummings interactions. In cavity QED with Rydberg atoms, this description is well established at the phenomenological and microscopic levels \cite{haroche1989cavity, haroche2006exploring}. Moreover, between the cavities, a suitably chosen laser field rotates the Rydberg atom in the desired direction \cite{saffman2010}, thereby implementing the unitary rotation required by our protocol.\\
A solid-state implementation of the same physics is provided by circuit quantum electrodynamics, or circuit QED. The Jaynes–Cummings (or quantum Rabi) Hamiltonian of atomic cavity QED governs the interaction between planar or three-dimensional microwave resonators, which offer single electromagnetic modes, and superconducting qubits, which operate as artificial two-level atoms. Since Blais et al.'s seminal proposal, in which the circuit-QED model replicates the canonical atom–cavity interaction with similar rotating-wave couplings and open-system channels driving cavity decay into electromagnetic baths \cite{blais2004}, this equivalence has been made clear. According to recent studies, circuit QED is an on-chip version of atomic cavity QED with the same light-matter Hamiltonians, controllable dissipation, and designed surroundings \cite{blais2021, GU20171}.


These findings demonstrate that, while retaining the same Hamiltonian structure, dissipative mechanisms, and sequential interaction paradigm as in atomic cavity QED, the cascaded cavity-QED model studied here can be replicated in a superconducting platform with artificial atoms and planar microwave cavities.

\section{Conclusions}
\label{conclusion}

In this work, we have developed a multitemperature estimation protocol based on a collisional model and analyzed its fundamental limitations and achievable precision. We identified a necessary and sufficient condition for the singularity of the Fisher information matrix (FIM) when thermal information about two baths is encoded in a single qubit probe. By incorporating a controlled unitary operation into the protocol, we demonstrated how this singularity can be lifted, thereby restoring a finite and attainable bound on the estimation error. Although simultaneous estimation of multiple temperatures was not possible owing to the non-commutativity of the symmetric logarithmic derivatives (SLDs), we have introduced strategies to improve the performance compared to the individual thermal Fisher information. 

Furthermore, we have determined the optimal interaction strength between the ancillary systems and the probes that maximizes the extractable joint information about the baths while maintaining a fundamental lower bound on precision. Our results also reveal that correlated ancillas can outperform uncorrelated ones provided that more than two ancillas are employed, thereby highlighting the role of multipartite correlations in quantum thermometry. 

Finally, we have proposed a scalable scheme for estimating the temperatures of $N$ independent baths by increasing the dimensionality of the ancillary systems, offering a pathway toward high-dimensional, multi-parameter quantum thermal metrology.

\begin{acknowledgments}
S.G. and S.C. thank C. Mukhopadhyay and A. Das for fruitful discussions. The authors acknowledge support by MUR (Ministero dell’Universit\`{a} e della Ricerca) through the PNRR Project ICON-Q – Partenariato Esteso NQSTI – PE00000023 – Spoke 2 – CUP: J13C22000680006. R.L.F. also acknowledges support by MUR through the PNRR Project QUANTIP – Partenariato Esteso NQSTI – PE00000023 – Spoke 9 – CUP: E63C22002180006. 
\end{acknowledgments}

\appendix
\section{Importance of $U_{rot}(\theta, \hat{r})$}
\label{appendix_u_rot}

Let us first consider the single run scenario where the first ancilla is interacting with two probes thermalized at temperatures $T_1$ and $T_2$. We will demostrate here that the controlled unitary rotation in between two encodings is important to avoid the interdependencies of the two encoded parameters i.e. getting $\text{det}\mathcal{F}_Q \neq 0$. 




Considering $H_{S_1}$ and $H_{A_1}$ to be the system and ancilla Hamiltonian, respectively, whereas the interaction Hamiltonian giving rise to the collision is given by $H_{int}^{S_1A_1}$.
We consider
\begin{align}
    H_{S_1} &= \frac{\omega_{S_1}}{2} \sigma _z \\
    H_{A_1} &= \frac{\omega_{A_1}}{2} \sigma _z \\
    H_{int}^{S_1A_1}&=g_1(\sigma_+\otimes \sigma_-+\sigma_-\otimes \sigma_+)\\
    &=g_1\big(\ket{01}\bra{10}+\ket{10}\bra{01}\big)
\end{align}
Thus, we have
\begin{align}
    U&_{S_1A_1}=e^{-i  H_{int}^{S_1A_1}t^1_{S_1A_1}}\nonumber\\
    &=\ket{00}\bra{00}+\ket{11}\bra{11}\nonumber\\
    &+\cos g_1t^1_{S_1A_1}\ket{01}\bra{01}-i\sin g_1t^1_{S_1A_1}\ket{01}\bra{10}\nonumber\\
    &-i\sin g_1 t^1_{S_1A_1}\ket{10}\bra{01}+\cos g_1 t^1_{S_1A_1}\ket{10}\bra{10}
\end{align}
Let us assume for simplicity $g_i=g$ and $t^i_{S_iA_i}=\tau$ for all $i$. Thus in the form of a matrix we get
\begin{align}
    U_{S_1A_1}=\begin{bmatrix}
        1 & 0 & 0 & 0 \\
        0 & \cos g\tau & -i\sin g\tau & 0 \\
        0 & i\sin g\tau & \cos g\tau & 0 \\
        0 & 0 & 0 & 1
    \end{bmatrix}
\end{align}
Now, notice that if we consider the ancilla state to be $\eta$ and the thermalized system state to be $\rho_{th}(T_1)$, the action of the collision is given by
\begin{align}
    \mathcal{E}_{T_1}[\eta]&=\text{Tr}_{S_1}\big[U_{S_1A_1}(\rho_{th}(T_1)\otimes\eta)U_{S_1A_1}^{\dagger}\big]\\
    &=\sum_{i=0,1}{}_{S_1}\bra{i} U _{S_1A_1} \Big(\sum_{j=0,1}\lambda_j\ket{j}\bra{j}\otimes \eta\Big)U_{S_1A_1}^{\dagger}\ket{i}_S \nonumber\\
    &=\sum_{i,j=0,1}\sqrt{\lambda_j}\bra{i}U_{S_1A_1}\ket{j} \, \eta \,\sqrt{\lambda_j}\bra{j}U_{S_1A_1}^{\dagger}\ket{i} \nonumber\\
    &= \sum_{i,j} K_{ij} \, \eta \, K_{ij}^{\dagger}
\end{align}
where
\begin{align}\label{kraus_formula}
    K_{ij} =\sqrt{\lambda_j}\bra{i}U_{S_1A_1}\ket{j}.
\end{align}
Also note, we considered here $\rho_{th}(T_1)=\sum_{i}\lambda_i(T_1)\ket{i}\bra{i}$ where $\lambda_i(T_1)=\mathcal{Z}(T_1)^{-1}e^{-E_i/{T_1}}$ where $\mathcal{Z}(T_1)$ is the partition function. Thus, we have
\begin{align}
K_{00}&=\sqrt{\lambda_0}\big[\ket{0}\bra{0}+\cos g\tau \ket{1}\bra{1} \big] \\
K_{01}&=-i\sin g\tau \sqrt{\lambda_1}\ket{1}\bra{0} \\
K_{10}&=i\sin g\tau \sqrt{\lambda_0}\ket{0}\bra{1} \\
K_{11}&=\sqrt{\lambda_1}\Big(\ket{1}\bra{1}+\cos g\tau \ket{0}\bra{0} \Big)
\end{align}
We can now write down the action of $\Lambda_T$ as
\begin{align}
    \mathcal{E}_{T_1}\Big[\ket{0}\bra{0}\Big] &= \begin{bmatrix}
        \lambda_0 + \lambda_1 \cos^2 g\tau & 0 \\
        0 & \lambda_1 \sin^2 g\tau
    \end{bmatrix} \\
    \mathcal{E}_{T_1}\Big[\ket{0}\bra{1}\Big] &= \begin{bmatrix}
        0 & \cos g\tau \\
        0 & 0
    \end{bmatrix} \\
    \mathcal{E}_{T_1}\Big[\ket{1}\bra{0}\Big] &= \begin{bmatrix}
        0 & 0 \\
        \cos g\tau & 0
    \end{bmatrix} \\
    \mathcal{E}_{T_1}\Big[\ket{1}\bra{1}\Big] &= \begin{bmatrix}
        \lambda_0 \sin^2 g\tau & 0 \\
        0 & \lambda_1+\lambda_0 \cos^2 g\tau
    \end{bmatrix} 
\end{align}
Thus in the vectorized form i.e. $\ket{i}\bra{j}\rightarrow \ket{ij}$, we have
\begin{align} \hat{\hat{\mathcal{E}}}_{T_1}=\begin{bmatrix}
        \lambda_0+\lambda_1 \cos^2 g\tau & 0 & 0 & \lambda_0 \sin^2 g\tau \\
        0 & \cos g\tau & 0 & 0 \\
        0 & 0 & \cos g\tau & 0 \\
        \lambda_1 \sin^2 g\tau & 0 & 0 & \lambda_1+\lambda_0 \cos^2 g\tau
    \end{bmatrix}
\end{align}
Now for two consecutive collisions on the ancilla by first and the second bath, we consider $p=\lambda_0(T_1)$ for $\hat{\hat{\mathcal{E}}}_{T_1}$ and $q=\lambda_0(T_2)$ for $\hat{\hat{\mathcal{E}}}_{T_2}$. Thus, the action of two consecutive collisions is given by 
\begin{align}
    \hat{\hat{\mathcal{E}}}_{f} &=  \hat{\hat{\mathcal{E}}}_{T_2}\cdot\hat{\hat{\mathcal{E}}}_{T_1} \nonumber \\
    &= \begin{bmatrix}
        1-u(T_1,T_2) & 0 & 0 & v(T_1,T_2) \\
        0 & \cos^2 g\tau & 0 & 0 \\
        0 & 0 & \cos^2 g\tau & 0 \\
        u(T_1,T_2) & 0 & 0 & 1-v(T_1,T_2) 
    \end{bmatrix}
\end{align}
where

\begin{align}
    u(T_1,T_2) &= \sin^2 g\tau \Big[1-p+(1-q) \cos^2 g\tau\Big] \\
    v(T_1,T_2) &= \sin^2 g\tau \Big[p+q \cos^2 g\tau\Big]
\end{align}
Now, if the initial ancilla state is $\rho_{A_1}=\ket{1}\bra{1}$, using the above, we get $\hat{\hat{\mathcal{E}_f[\rho_{A_1}]}}=\hat{\hat{\mathcal{E}}}_f\cdot \hat{\hat{\rho_{A_1}}}$, so that
\begin{equation}
    \rho^f_{A_1}=\mathcal{E}_f[\rho_{A_1}] = \begin{bmatrix}
        v(T_1,T_2) & 0 \\
        0 & 1-v(T_1,T_2)
    \end{bmatrix}.
\end{equation}
Thus, the SLD operators $L_1$ and $L_2$ are given by
\begin{align}
    L_i= \sum_{j,k=0,1}\frac{2\bra{\alpha_j}\frac{\partial \rho^f_{A_1}}{\partial T_i}\ket{\alpha_k}}{\alpha_j+\alpha_k} \ket{\alpha_j}\bra{\alpha_k},
\end{align}
where we have the eigenvalue decomposition $\rho^f_{A_1}=\alpha_1\ket{\alpha_1}\bra{\alpha_1}+\alpha_2\ket{\alpha_2}\bra{\alpha_2}$. 
Hence, for this case, we have
\begin{align}
    L_i &= \big(\partial_i v\big)\Big(\frac{1}{v}\ket{0}\bra{0}-\frac{1}{1-v}\ket{1}\bra{1}\Big) \\
    \Tr&[\eta_f L_iL_j] = \frac{\big(\partial_i v\big)\big(\partial_j v\big)}{v(1-v)}
\end{align}
where we used the shorthand notation $\partial_i$ for $\frac{\partial}{\partial T_i}$.

Therefore, using this, we get
\begin{align}
    \mathcal{F}_{Q}(\eta_f)=\frac{1}{v(1-v)}\begin{bmatrix}
        (\partial_1 v)^2 & (\partial_1 v)(\partial_2 v) \\
        (\partial_2 v)(\partial_1 v) & (\partial_2 v)^2
    \end{bmatrix} 
\end{align}
As can be easily seen here that $\text{det}(\mathcal{F}_{Q})=0$ which would imply that the protocol automatically makes the two parameters dependent. In order to avoid this issue, we introduce a rotation $U_{rot}(\theta,\hat{r})$ between two successive collisions as shown in Fig. \ref{schematic}. The action of the rotation in the vectorized form is given by
\begin{align}
    \hat{\hat{\mathcal{U}}}_{rot}=\frac{1}{2}\begin{bmatrix}
        1 & i & -i & 1 \\
        i & 1 & 1 & -i \\
        -i & 1 & 1 & i \\
        1 & -i & i & 1
    \end{bmatrix}
\end{align}
Thus, $\hat{\hat{\mathcal{E}}}_f$ is given by
\begin{widetext}
\begin{align}
    &\hat{\hat{\mathcal{E}}}_{f} =  \hat{\hat{\mathcal{E}}}_{T_2}\cdot\hat{\hat{\mathcal{U}}}_{rot}\cdot\hat{\hat{\mathcal{E}}}_{T_1} = \begin{bmatrix}
    \mu(q)
    & \tfrac{i}{2}\,\zeta(g)\cos g
    & -\tfrac{i}{2}\,\zeta(g)\cos g
    & \mu(q) \\[6pt]
    i\chi(1 - p)
    & \tfrac{1}{2}\,\zeta(g)
    & \tfrac{1}{2}\,\zeta(g)
    & -i\chi(p) \\[6pt]
    -i\chi(1 - p)
    & \tfrac{1}{2}\,\zeta(g)
    & \tfrac{1}{2}\,\zeta(g)
    & i\chi(p) \\[6pt]
    1 - \mu(q)
    & -\tfrac{i}{2}\,\zeta(g)\cos g
    & \tfrac{i}{2}\,\zeta(g)\cos g
    & 1 - \mu(q)
\end{bmatrix}
\end{align}
\end{widetext}
where
\begin{align}
    \mu(q) &= \tfrac{1}{4}\!\big[(1 + 2q) + (1 - 2q)\cos 2g\big]\\
    \chi(p) &= \tfrac{1}{2} \big[(1 - p) + p\cos 2g\big]\cos g \\
    \zeta(g) &=\cos^2 g
\end{align}
Thus, one can easily see that $\rho^f_{A_1}=\Lambda_f[\rho_{A_1}]=\hat{\hat{\Lambda}}_f\cdot\hat{\hat{\rho_{A_1}}}$ is in the form
\begin{equation}
    \rho^f_{A_1}=\begin{bmatrix}
    \mu (q) & -i\chi(p) \\
    i\chi(p) & 1-\mu(q)
    \end{bmatrix}
\end{equation}
We will now show that $\text{det}(\mathcal{F}_{joint})\neq 0$ as long as $\rho^f_{A_1}$ is full rank. Assuming $\rho^f_{A_1}$ to be full rank we have 
\begin{equation}\label{anc_final_spectral} \rho^f_{A_1}=\alpha_0\ket{\alpha_0}\bra{\alpha_0}+\alpha_1\ket{\alpha_1}\bra{\alpha_1}
\end{equation} 
where
\begin{align}
    \ket{\alpha_k}&=\tfrac{1}{\sqrt{1+\beta_k^2}}\begin{bmatrix}
        1 \\ i \beta_k
    \end{bmatrix}\\
    \beta_k&= \tfrac{\alpha_k-\mu}{\chi}\,\,;\,\,
    \alpha_{0/1}=\tfrac{1}{2}\bigg[1\pm\sqrt{1-4\text{det}(\rho^f_{A_1})}\bigg]
\end{align}
Also, note we get 
\begin{align}
  \tfrac{\partial}{\partial T_1}\big(\rho^f_{A_1}\big) &=\dot{\chi}\sigma_y  \\
  \tfrac{\partial}{\partial T_2}\big(\rho^f_{A_1}\big) &=\dot{\mu}\sigma_z
\end{align}
 where $\dot{\mu}=d\mu/ d T_2$ and $\dot{\chi}=d\chi/ d T_1$. Thus, using Theorem \ref{theorem1} we can see that $\text{det}(\mathcal{F}_Q)\neq 0$. Also, we get
\begin{widetext}
    \begin{align}
    L_1&=\dot{\chi}\bigg\{\dfrac{2\beta_0}{\alpha_0 (1 + \beta_0^2)}\ket{\alpha_0}\bra{\alpha_0}+\dfrac{2\beta_1}{\alpha_1 (1 + \beta_1^2)}\ket{\alpha_1}\bra{\alpha_1}
    +\dfrac{2(\beta_0 + \beta_1)}{\sqrt{(1 + \beta_0^2)(1 + \beta_1^2)}} \Big(\ket{\alpha_0}\bra{\alpha_1}+\ket{\alpha_1}\bra{\alpha_0}\Big)\bigg\} \label{protocol_L_1_form}\\
    L_2&= \dot{\mu} \bigg\{ \dfrac{1 - \beta_0^2}{\alpha_0 (1 + \beta_0^2)}\ket{\alpha_0}\bra{\alpha_0} +  \dfrac{1 - \beta_1^2}{\alpha_1 (1 + \beta_1^2)} \ket{\alpha_1}\bra{\alpha_1} + \dfrac{2(1 - \beta_0 \beta_1)}{\sqrt{(1 + \beta_0^2)(1 + \beta_1^2)}}\Big(\ket{\alpha_0}\bra{\alpha_1}+\ket{\alpha_1}\bra{\alpha_0}\Big)\bigg\} \label{protocol_L_2_form}
\end{align}
\end{widetext}

It can be easily seen from above that $[L_1,L_2]\neq0$ as long as $\beta_0\neq\beta_1$, i.e., $\rho^f_{A_1}$ is full rank.

\clearpage
\bibliography{biblio_corr_bounds.bib}

\end{document}